\documentclass[conference]{IEEEtran}
\usepackage{amsfonts}
\usepackage{amsmath}
\usepackage{graphicx}
\usepackage{url}
\usepackage{eepic, epsfig, amsmath, amssymb, latexsym, setspace, subfig}
\usepackage{rotating}
\usepackage{amsfonts}
\usepackage{amsmath}
\usepackage{graphicx}
\usepackage{nicefrac}
\usepackage{lipsum,multicol}
\usepackage{algorithm2e}
\usepackage{authblk}

\newtheorem{theorem}{Theorem}[section]

\newtheorem{lemma}[theorem]{Lemma}

\numberwithin{equation}{section}

\newcommand{\qed}{\rule{7pt}{7pt}}
\newenvironment{proof}{\noindent{\bf Proof}\hspace*{1em}}{\hfill\qed\vspace{0.125in}}

\begin{document}
%
\title{Compressed Hypothesis Testing: to Mix or Not to Mix?}


\author[1]{Weiyu Xu}
\author[2]{Lifeng Lai}

\affil[1]{Department of ECE\\
University of Iowa}
\affil[2]{Department of Electrical and Computer Engineering\\
Worcester Polytechnic Institute}

\maketitle

\begin{abstract}
In this paper, we study the hypothesis testing problem of, among $n$ random variables, determining $k$ random variables which have different probability distributions from the rest $(n-k)$ random variables. Instead of using separate measurements of each individual random variable, we propose to use mixed measurements which are functions of  multiple random variables. It is demonstrated that $O({\displaystyle \frac{k \log(n)}{\min_{P_i, P_j} C(P_i, P_j)}})$ observations are sufficient for correctly identifying the $k$ anomalous random variables with high probability, where $C(P_i, P_j)$ is the Chernoff information between two possible distributions $P_i$ and $P_j$ for the proposed mixed observations. We characterized the Chernoff information respectively under fixed time-invariant mixed observations, random time-varying mixed observations, and deterministic time-varying mixed observations; in our derivations, we introduced the \emph{inner and outer conditional Chernoff information} for time-varying measurements. It is demonstrated that mixed observations can strictly improve the error exponent of hypothesis testing, over separate observations of individual random variables. We also characterized the optimal mixed observations maximizing the error exponent, and derived an explicit construction of the optimal mixed observations for the case of Gaussian random variables. These results imply that mixed observations of random variables can reduce the number of required samples in hypothesis testing applications. Compared with compressed sensing problems, this paper considers random variables which are allowed to dramatically change values in different measurements.
\end{abstract}

\IEEEpeerreviewmaketitle

\section{Introduction}

 In many areas of science and engineering, one needs to infer statistical information of objects of interest. Here, the statistical information of interest can be the mean, variance or even distributions of certain random variables \cite{quickestdetection, DuffieldProbabilistic, ENV:ENV265,  Pisarenko19874,PrestiProbabilistic, TseLink}. In fact, inferring distributions of random variables are essential in anomaly detections, for example, quickest detections of potential hazards or changes, \cite{quickestdetection, basseville, LaiCognitive}. In \cite{quickestdetection, LifengMultiple, MalloyAsilomar,MalloyISIT,MalloyRare, LaiCognitive}, the authors are interested in knowing the probability distribution of $n$ independent random variables $X_i$, $1\leq i \leq n$, which follow either probability distribution $f_1(\cdot)$ or $f_2(\cdot)$. In order to get the probability distribution information about the $n$ random variables, one can sample each random variable separately multiple times, and then infer probability distribution information from their separate samples. However, when the number of random variables $n$ goes large, the requirement on sampling rates and sensing resources can be tremendous. In some applications, due to physical constraints \cite{DuffieldProbabilistic, PrestiProbabilistic, TseLink}, we cannot even directly get separate samples of individual random variables.  This raises the question of whether we can infer statistical information of interest from a much smaller number of samples.

Fortunately, in some applications such as network tomography \cite{DuffieldProbabilistic, TseLink} and cognitive radio \cite{LaiCognitive}, these $n$ random variables can be described by a parsimonious model. For example, suppose only $k$ ($k \ll n$) out of $n$ random variables take the probability distribution $f_2(\cdot)$ while the much larger set of remaining $(n-k)$ random variables take the probability distribution $f_1(\cdot)$ \cite{MalloyAsilomar}. Again, one natural way to find out the $k$ anomalous random variables is to do one-by-one hypothesis testing for these $n$ random variables. One can get $l$ samples for each random variable $X_i$, $1\leq i \leq n$, and then use existing hypothesis testing techniques to determine whether $X_i$ follows the probability distribution $f_1(\cdot)$ or $f_2(\cdot)$. Thus, to ensure correctly identifying the $k$ anomalous random variables with high probability, at least $\Theta(n)$ samples are needed for one-by-one hypothesis testing. This inevitably creates an enormous burden on data collecting and processing, especially when $n$ is very large, and the number of sensing samples is very small.

In this paper, since $k \ll n$, inspired by compressed sensing \cite{CT1,Neighborlypolytope,DonohoTanner}, we propose to use \emph{non-adaptive} \emph{\emph{mixed}} measurements (or observations, which we use interchangeably with measurements) of $n$ random variables, instead of separate sampling of these $n$ random variables individually, to quickly detect the $k$ anomalous random variables. The basic idea is to make each sample a function of $n$ random variables, instead of a realization of an \emph{individual} random variables. In this paper, we consider three different types of mixed observations: \emph{fixed time-invariant} mixed measurements, \emph{random time-varying} measurements, and \emph{deterministic time-varying} measurements. For these different types of mixed observations, we characterize the number of measurements needed to achieve a specified hypothesis testing error probability. Our analysis has shown that there is an advantage of performing this method in reducing the number of required samples for reliable hypothesis testing. The number of samples needed to correctly identify the $k$ anomalous random variables can be reduced to $O({\displaystyle \frac{k \log(n)}{\min_{P_i, P_j} C(P_i, P_j)}})$ observations, where $C(P_i, P_j)$ is the Chernoff information between two possible distributions $P_i$ and $P_j$ for the proposed mixed observations.

In this paper, we have also shown that mixed observations can strictly reduce error exponent of hypothesis testing, compared with separate sampling of individual random variables. Under the special cases of Gaussian random variables, optimal mixed measurements are derived to maximize the error exponent of hypothesis testing.

In \cite{MalloyAsilomar, MalloyISIT, MalloyRare}, the authors have considered the same problem setting of finding the $k$ anomalous random variables among $n$ random variables. By utilizing the sparsity of anomalous random variables, \cite{MalloyAsilomar, MalloyISIT, MalloyRare} optimized adaptive samplings of individual random variables, and have reduced the number of needed samples for individual random variables. Compared with this paper, \cite{MalloyAsilomar, MalloyISIT, MalloyRare} have considered adaptive observations of individual random variables, instead of \emph{\emph{non-adaptive}} \emph{mixed} observations in this paper. It is noted that the total number of observations is at least $\Theta(n)$ \cite{MalloyAsilomar, MalloyISIT, MalloyRare}, if one is restricted to sample the $n$ random variables individually. In a related work of ours, we have considered \emph{adaptive} mixed observations in quickest detection and search problem \cite{JunISIT} to speed up finding one anomalous random variable.

This paper is organized as follows. In Section \ref{sec:model}, we introduce the mathematical models for considered problems and mixed observations. In Section \ref{sec:fixed_time_invariant}, we investigate the hypothesis testing problem using time-invariant mixed observations; and propose corresponding hypothesis testing algorithms and analysis. In Section \ref{sec:random_timevarying}, we consider using random time-varying mixed observations to identify the anomalous random variables. In Section \ref{sec:deterministic_timevarying}, we consider using deterministic time-varying mixed observations for hypothesis testing, and derive a bound on the error probability. In Section \ref{sec:example}, we demonstrate, by examples of Gaussian random variable vectors, that linear mixed observations can strictly improve the error exponent in hypothesis testing. In Section \ref{sec:optimalmixing}, we derive the optimal mixed measurements for Gaussian random variables. In Section \ref{sec:simulation}, we provide numerical simulation results. Section \ref{sec:conclusion} concludes this paper.

\section{Mathematical Models}
\label{sec:model}
We consider $n$ i.i.d. random variables $X_1$, $X_2$,..., $X_n$. Out of these $n$ random variables, $(n-k)$ of them follow a known distribution $f_1(\cdot)$; while the other $k$ random variables follow another known distribution $f_2(\cdot)$, where $k \ll n$.  However, it is unknown which $k$ random variables follow the distribution $f_2(\cdot)$. Our objective is then to identify these $k$ anomalous random variables, with as few samples as possible.

We take $m$ mixed observations of the $n$ random variables at $m$ time indices: at time index $t=j$, $1\leq j \leq m$, the measurement result $$Y_j=g_j(X_1^j, X_2^j,..., X_n^j),$$ is a function of the realizations of the $n$ random variables at time $j$. In this paper, we only consider the case when the functions $g_j$ are \emph{linear}. When the functions $g_j$, $1\leq j \leq m$, are the \emph{same} for each $1\leq j \leq m$, we simply denote them by $g(\cdot)$. Thus for $1\leq j \leq m$, the $j$-th measurement  $$Y_j=g_j(X_1^j, X_2^j,..., X_n^j)=\sum_{i=1}^{n}a_i^j X_i^j,$$ where $a_i^j$ is a real number.  In this paper, each random variable $X_i$ takes an \emph{independent} realization in each measurement; while in the now well-known compressed sensing problem, for each measurement, the variables of interest, denoted by $(x_1,x_2,x_3,...,x_n)$, are \emph{deterministic} values,  and basically take the same values in each measurement. In some sense, our problem is a probabilistic generalization of the compressed sensing problems.

In the compressed sensing literature, Bayesian compressed sensing \cite{BayesianCS1, StatisticalCS2} stands out as one model where prior probability distribution of the vector $(x_1,x_2,x_3,...,x_n)$ is considered. However, in \cite{BayesianCS1, StatisticalCS2}, the values of $(x_1,x_2,x_3,...,x_n)$ are \emph{fixed} once the random variables are realized from the prior probability distribution, and then generally remain \emph{unchanged} across different measurements. That is fundamentally different from our setting where the random variables dramatically change over different measurements. In the compressed sensing literature, there are very interesting works discussing compressed sensing for \emph{smoothly} time-varying signals \cite{timevaryingGiannakis, Vaswani, DavidTimeVarying}. In contrast, objects of interest in this research are random variables taking \emph{completely independent} realizations at different time indices; and, thus, we are interested in recovering statistical information of random variables, rather than recovering the deterministic values.

\section{Taking Fixed Time-Invariant Measurements}
\label{sec:fixed_time_invariant}
In this section, we consider mixed measurements which are time-invariant across different measurements. We first give the likelihood ratio test algorithm, namely Algorithm \ref{alg:likelihoodratiotest_timeinvariant}, over the possible $\binom{n}{k}$ hypothesis. Then we analyze the number of needed samples through another suboptimal algorithm \ref{alg:pairwise_timeinvariant}.

\subsection{Algorithm}

\begin{algorithm}
 \SetAlgoLined
 \KwData{observation data $Y_1$, $Y_2$, ..., $Y_m$}
 \KwResult{$k$ anomalous random variables}

\begin{itemize}
\item For each hypothesis $H_i$ ($1\leq i\leq L$), calculate the likelihood $P(Y|H_i)$:
\item Choose the hypothesis with the maximum likelihood.
\item Decide the corresponding $k$ random variables as the anomalous random variables.
\end{itemize}

 \caption{Likelihood Ratio Test from Deterministic Time Invariant Measurements}
 \label{alg:likelihoodratiotest_timeinvariant}
\end{algorithm}

To analyze the number of needed samples for achieving a certain hypothesis testing error probability, we consider another hypothesis testing Algorithm \ref{alg:pairwise_timeinvariant} based on pairwise hypothesis testing, which is suboptimal compared to the likelihood ratio test algorithm. There are $L=\binom{n}{k}$ possible probability distributions for the output of the function $g(\cdot)$, depending on which $k$ random variables are anomalous. We denote these possible probability distributions as $P_1$, $P_2$, ..., $P_L$. Our simple algorithm is to find the true distribution by doing pairwise Neyman-Pearson hypothesis testing \cite{coverbook} of these $L$ probability distributions. We denote the $m$ observations by $Y_1$, $Y_2$, ..., $Y_m$.

\begin{algorithm}
 \SetAlgoLined
 \KwData{observation data $Y_1$, $Y_2$, ..., $Y_m$}
 \KwResult{$k$ anomalous random variables}
\begin{itemize}
\item
For all pairs of distinct probability distributions $P_i$ and $P_j$ ($1\leq i, j \leq L$ and $i\neq j$), perform Neyman-Pearson testing for two hypothesis:
\begin{itemize}
\item $Y_1$, $Y_2$, ..., $Y_m$ follow probability distribution $P_i$
\item $Y_1$, $Y_2$, ..., $Y_m$ follow probability distribution $P_j$
\end{itemize}

\item

  \eIf{there exists a certain $j^*$, $P_{j^*}$ is the winning probability distribution whenever it is involved in a pairwise hypothesis testing, }{
   declare the $k$ random variables producing $P_{j^*}$ as anomalous random variables\;
   }{ declare a failure in finding the $k$ anomalous random variables.
  }

\end{itemize}

\caption{Hypothesis Testing from Time-invariant Mixed Measurements}
\label{alg:pairwise_timeinvariant}
\end{algorithm}

\subsection{Number of Samples}
\begin{theorem}
Consider time-invariant fixed observations $Y$ for $n$ random variables $X_1$, $X_2$, ..., and $X_n$. With $O(\frac{k \log(n)}{\min_{1\leq i,j \leq L, i\neq j} C(P_i, P_j)}$ mixed measurement, with high probability, Algorithms \ref{alg:likelihoodratiotest_timeinvariant} and \ref{alg:pairwise_timeinvariant} correctly identify the $k$ anomalous random variables.
Here $L$ is the number of hypothesis, $P_i$, $1\leq i \leq n$ is the output probability distribution for measurements $Y$ under hypothesis $H_i$, and
\begin{equation*}
C(P_i, P_j)=-\min_{0\leq \lambda \leq 1} \log\left(\int{P_i^{\lambda}(x)P_j^{1-\lambda}(x)dx}\right)
\end{equation*}
is the Chernoff information between two distributions $P_i$ and $P_j$.
\label{thm:numberofsamples}
\end{theorem}

\begin{proof}
In Algorithm \ref{alg:pairwise_timeinvariant}, for two probability distributions $P_i$ and $P_j$, we choose the probability likelihood ratio threshold of the Neyman-Pearson testing in such a way that the error probability decreases with the largest possible error exponent, namely the Chernoff information between $P_i$ and $P_j$
\begin{equation*}
C(P_i, P_j)=-\min_{0\leq \lambda \leq 1} \log\left(\int{P_i^{\lambda}(x)P_j^{1-\lambda}(x)dx}\right)
\end{equation*}

So overall, the smallest possible error exponent of making a error between any pair of probability distributions is
\begin{equation*}
E=\min_{1\leq i,j \leq L, i\neq j} C(P_i, P_j).
\end{equation*}

Without loss of generality, we assume that $P_1$ is the true probability distribution for the observation data $Y$. Since the error probability $P_{e}$ in the Neyman-Pearson testing scales like $P_{e}\doteq2^{-m C(P_i, P_j)} \leq 2^{-mE}$, by a union bound over the $L-1$ possible pairs $(P_1, P_j)$, the probability that $P_1$ is not correctly identified as the true probability distribution scales at most as $L2^{-mE}$. So $\Theta(k \log(n)E^{-1})$ samplings are enough for identifying the $k$ anomalous samples with high probability.

\end{proof}

When $E$ grows polynomially with $n$, this implies a significant reduction in the number of samples needed. Consider a simple example where $k=1$, $f_1(x)\sim \delta(x)$ is the Dirac delta function, and $f_2(x)\sim N(0,1)$ is the standard Gaussian distribution. Let $P_j$, $1\leq j \leq n$, be the sketch output distribution corresponding to the case where the $j$-th random variable follows distribution $ N(0,1)$. Then for $1\leq i,j \leq L$, $i\neq j$, using the Chernoff information results for two Gaussian random variables \cite{ChernoffInformation},
\begin{equation*}
C(P_i,P_j)= \max_{0\leq \beta \leq 1}\frac{1}{2}\log\frac{\beta a_i^2+(1-\beta)a_j^2}{a_i^{2\beta}a_j^{2(1-\beta)}}\geq \frac{1}{2}\log \frac{a_i^2+a_j^2}{{2a_i a_j}}.
\end{equation*}

So as long as the ratio $\frac{\max\{\alpha_i, \alpha_j\}}{\min\{\alpha_i, \alpha_j\}}$ for any two coefficients $\alpha_i$ and $\alpha_j$ is always larger than a constant $\gamma>2$, we just need $O(\log(n))$ measurements to find out which random variable is abnormal. If we are allowed to use time-varying non-adaptive sketching functions or are allowed to design adaptive measurements based on the history of measurement results, we may need fewer samples. In the next section, we discuss the performance of time-varying non-adaptive mixed measurements for this problem.

\section{Taking Random Time-Varying Measurements}
\label{sec:random_timevarying}

In this section, we consider the same problem setup as in Section \ref{sec:model}, except that each measurement is the inner product between $X$ and one independent realization $(a_1^j,a_2^j,..., a_{n}^j)$ of a random vector $A$. Namely, each observation is given by
$$Y_j=<A^j, X>=\sum\limits_{i=1}^{n}a_i^j X_i, 1\leq j \leq m,$$
where $A^j=(a_1^j,a_2^j,..., a_{n}^j)$.
We assume that the random vector $A$ has a probability density function $P(A)$, and, the realizations $A^j$'s of $A$ are independent across different measurements.  We consider random time-varying measurements, because, inspired compressed sensing, random measurements often given desirable performance \cite{CT1, DonohoTanner}.

Under this setup, we would like to design hypothesis testing algorithms, to decide which probability distribution the vector $X$ is following among the $\binom{n}{k}$ hypothesis; moreover, we are interested in analyzing the error probability of such hypothesis testing algorithms.

\subsection{Hypothesis Testing from Random Time-Varying Measurements}
We first give the likelihood ratio test algorithm over the possible $\binom{n}{k}$ hypothesis.

\begin{algorithm}
 \SetAlgoLined
 \KwData{observation data $(A^1, Y_1)$, $(A^2,Y_2)$, ..., $(A^m,Y_m)$}
 \KwResult{$k$ anomalous random variables}

\begin{itemize}
\item For each hypothesis $H_i$ ($1\leq i\leq L$), calculate the likelihood $P(A,Y|H_i)$:
\item Choose the hypothesis with the maximum likelihood.
\item Decide the corresponding $k$ random variables as the anomalous random variables.
\end{itemize}

 \caption{Likelihood Ratio Test from Random Time-Varying Measurements}
 \label{alg:likelihoodratiotest_timevarying}
\end{algorithm}

For the purpose of analyzing the error probability of likelihood ratio test, we further propose one hypothesis testing algorithm based on pairwise comparison.

\begin{algorithm}
 \SetAlgoLined
 \KwData{observation data $(A^1, Y_1)$, $(A^2,Y_2)$, ..., $(A^m,Y_m)$}
 \KwResult{$k$ anomalous random variables}
\begin{itemize}
\item For all pairs of hypothesis $H_i$ and $H_j$ ($1\leq i, j \leq L$ and $i\neq j$), perform Neyman-Pearson testing of the following two hypothesis:
\begin{itemize}
\item $(A^1,Y_1)$, $(A^2,Y_2)$, ..., $(A^m, Y_m)$ follow the probability distribution $P(A,Y|H_i)$;
\item $(A^1,Y_1)$, $(A^2,Y_2)$, ..., $(A^m, Y_m)$ follow probability distribution $P(A,Y|H_j)$.
\end{itemize}

\item
  \eIf{there exists a certain $j^*$, such that $H_{j^*}$ is the winning hypothesis, whenever it is involved in a pairwise hypothesis testing, }{
   declare the $k$ random variables producing $H_{j^*}$ as anomalous random variables\;
   }{ declare a failure in finding the $k$ anomalous random variables.
  }

\end{itemize}

 \caption{Hypothesis Testing from Random Time-Varying Measurements}
 \label{alg:timevaryingalg2}
\end{algorithm}

\subsection{Number of Samples for Random Time-Varying Measurements}

\begin{theorem}
Consider time variant random observations $Y_j$, $1\leq j \leq m$, for $n$ random variables $X_1$, $X_2$, ..., and $X_n$. With $O(\frac{k \log(n)}{\min\limits_{1\leq i,j \leq L, i\neq j} IC(P_{Y|A,H_i}, P_{Y|A,H_j})})$ random time-varying measurements, with high probability, Algorithms \label{alg:likelihoodratiotest_timevarying} and  \label{alg:likelihoodratiotest_timevarying} correctly identify the $k$ anomalous random variables.
Here $L$ is the number of hypothesis, $P_{Y|A,H_i}$, $1\leq i \leq n$ is the output probability distribution for measurements $Y$ under hypothesis $H_i$ and measurements $A$; and
\begin{eqnarray*}
&&IC(P_{Y|A,H_i}, P_{Y|A,H_j})\\
&&=-\min_{0\leq \lambda \leq 1} \log\left(\int P(A)\right.\\
&&~~~~~~~~~~~~~~~~~~~~\left.(\int{P_{Y|A,H_i}^{\lambda}(x)P_{Y|A,H_j}^{1-\lambda}(x)\,dx}) \,dA\right)\\
&&=-\min_{0\leq \lambda \leq 1} \log\left( E_{A} \left(\int{P_{Y|A,H_i}^{\lambda}(x)P_{Y|A,H_j}^{1-\lambda}(x)\,dx}\right)\right)
\end{eqnarray*}
is the \emph{inner conditional Chernoff information} between two hypothesis for observations $Y$, conditioned on the probability distribution of time-varying measurements $A$.
\label{thm:numberofsamples}
\end{theorem}

\begin{proof}
In Algorithm \ref{alg:timevaryingalg2}, for two different hypothesis $H_i$ and $H_j$, we choose the probability likelihood ratio threshold of the Neyman-Pearson testing in a way, such that the hypothesis testing error probability decreases with the largest error exponent, namely the Chernoff information between $P(A,Y|H_i)$ and $P(A,Y|H_j)$
\begin{eqnarray*}
&&IC(P(A,Y|H_i), P(A,Y|H_j))\\
&&=-\min_{0\leq \lambda \leq 1} \log\left(\int{P(A,Y|H_i)^{\lambda}(x)P(A,Y|H_j)^{1-\lambda}(x)dx}\right)
\end{eqnarray*}

Since the random tim-varying measurements are independent of random samples $X$ and the hypothesis $H_i$ or $H_j$,
$$P(A,Y|H_i)=P(A|H_i) P(Y|H_i,A)=P(A) P(Y|H_i,A),$$
$$P(A,Y|H_j)=P(A|H_j) P(Y|H_j,A)=P(A) P(Y|H_j,A).$$
Then the Chernoff information is simplified to
\begin{eqnarray*}
&&IC(P_{Y|A,H_i}, P_{Y|A,H_j})\\
&&=-\min_{0\leq \lambda \leq 1} \log\left(\int P(A)\right.\\
&&~~~~~~~~~~~~~~~~~~~~\left.(\int{P_{Y|A,H_i}^{\lambda}(x)P_{Y|A,H_j}^{1-\lambda}(x)\,dx}) \,dA\right)\\
&&=-\min_{0\leq \lambda \leq 1} \log\left( E_{A} \left(\int{P_{Y|A,H_i}^{\lambda}(x)P_{Y|A,H_j}^{1-\lambda}(x)\,dx}\right)\right)
\end{eqnarray*}

By the Holder's inequality, we have
\begin{eqnarray*}
&&IC(P_{A,Y|H_i}, P_{A,Y|H_j})\\
&&\geq -\min_{A}\log\left(1-P(A)+P(A)e^{-C(P_{Y|A, H_i}, P_{Y|A, H_j})} \right),\\
\end{eqnarray*}
where
\begin{eqnarray*}
&&C(P_{Y|A, H_i}, P_{Y|A, H_j})\\
&&=-\min_{0\leq \lambda \leq 1} \log\left(\int{P(Y|A,H_i)^{\lambda}(x)P(Y|A,H_j)^{1-\lambda}(x)dx}\right)
\end{eqnarray*}
is the ordinary Chernoff information between $P_{Y|A, H_i}$, and $P_{Y|A, H_j}$. So as long there exits measurements $A$ of a positive probability, such that the ordinary Chernoff information
is positive, then the inner condition Chernoff information $IC(P_{A,Y|H_i}, P_{A,Y|H_j})$ will also be positive.

Overall, the smallest possible error exponent between any pair of hypothesis is
\begin{equation*}
E=\min_{1\leq i,j \leq L, i\neq j} C(P_{A,Y|H_i}, P_{A,Y|H_j}).
\end{equation*}

Without loss of generality, we assume $H_1$ is the true hypothesis. Since the error probability $P_{e}$ in the Neyman-Pearson testing is
$$P_{e}\doteq2^{-m C(P_{A,Y|H_i}, P_{A,Y|H_j})} \leq 2^{-mE}.$$ By a union bound over the $L-1$ possible pairs $(H_1, H_j)$, the probability that $H_1$ is not correctly identified as the true hypothesis is upper bounded by $L2^{-mE}$ in terms of scaling. So $m=\Theta(k \log(n)E^{-1})$ samplings are enough for identifying the $k$ anomalous samples with high probability.  When $E$ grows polynomially with $n$, this implies a significant reduction in the number of needed samples.

\end{proof}

\section{Taking Deterministic Time-Varying Measurements}
 \label{sec:deterministic_timevarying}

In this section, we consider mixed measurements which are allowed to vary over time. However, each measurement is predetermined, so that exactly $p(A)m$ (assuming that $p(A)m$ are integers) measurements use the measurement $A$. In contrast, in random time-varying measurements, each measurement is taken as $A$ with probability $p(A)$, and thus the number of measurements using $A$ is a random variable.

\subsection{Algorithms}
 In deterministic time-varying measurements, we first give the likelihood ratio test algorithm over the possible $\binom{n}{k}$ hypothesis.
\begin{algorithm}
 \SetAlgoLined
 \KwData{observation data $Y_1$, $Y_2$, ..., $Y_m$}
 \KwData{deterministic measurements $A^1$, $A^2$, ..., $A^m$}
 \KwResult{$k$ anomalous random variables}

\begin{itemize}
\item For each hypothesis $H_i$ ($1\leq i\leq L$), calculate the likelihood $P(Y|A,H_i)$:
\item Choose the hypothesis with the maximum likelihood.
\item Decide the corresponding $k$ random variables as the anomalous random variables.
\end{itemize}

 \caption{Likelihood Ratio Test from Deterministic Time-Varying Measurements}
 \label{alg:likelihoodratiotestdeterministic}
\end{algorithm}

 Similar to the analysis of random time-varying measurements, for the purpose of analyzing the error probability, we consider one hypothesis testing algorithm based on pairwise comparison.
\begin{algorithm}
 \SetAlgoLined
 \KwData{observation data $Y_1$, $Y_2$, ..., $Y_m$}
 \KwData{deterministic measurements $A^1$, $A^2$, ..., $A^m$}
 \KwResult{$k$ anomalous random variables}
\begin{itemize}
\item For all pairs of hypothesis $H_i$ and $H_j$ ($1\leq i, j \leq L$ and $i\neq j$), perform Neyman-Pearson testing of the following two hypothesis:
\begin{itemize}
\item $Y_1$, $Y_2$, ..., $Y_m$ follow the probability distribution $P(Y|H_i,A)$;
\item $Y_1$, $Y_2$, ..., $Y_m$ follow probability distribution $P(Y|H_j,A)$.
\end{itemize}

\item
  \eIf{there exists a certain $j^*$, such that $H_{j^*}$ is the winning hypothesis, whenever it is involved in a pairwise hypothesis testing, }{
   declare the $k$ random variables producing $H_{j^*}$ as anomalous random variables\;
   }{ declare a failure in finding the $k$ anomalous random variables.
  }

\end{itemize}

 \caption{Hypothesis Testing from Deterministic Time-Varying Measurements}
 \label{alg:timevaryingalg2deterministic}
\end{algorithm}

\subsection{Number of Samples for Deterministic Time-Varying Measurements}
\begin{theorem}
Consider time-varying deterministic observations $Y_j$, $1\leq j \leq m$, for $n$ random variables $X_1$, $X_2$, ..., and $X_n$. $L$ is the number of hypothesis for the distribution of the vector $(X_1,X_2,...,X_n)$.

For $\lambda\in[0,1]$ and two hypothesis $H_i$ and $H_j$ ($1\leq i ,j \leq L$), define
$$P_{\lambda}(Y|A)=\frac{P^\lambda(Y|A,H_i)P^{1-\lambda}(Y|A,H_j)}{\sum\limits_{Y}P^\lambda(Y|A,H_i)P^{1-\lambda}(Y|A,H_j)}$$

\begin{equation*}
Q_{\lambda, i\rightarrow j}=E_{A} \left\{D \left (     P_{\lambda}(Y|A)~||~P(Y|A,H_i)                \right ) \right\}
\end{equation*}
\begin{equation*}
Q_{\lambda, j\rightarrow i}=E_{A} \left\{D \left (     P_{\lambda}(Y|A)~||~P(Y|A,H_j)                \right ) \right\}
\end{equation*}

Furthermore, we define the \emph{outer conditional Chernoff information} $OC(P_{Y|A,H_i}, P_{Y|A,H_j})$ between $H_i$ and $H_j$, under deterministic time-varying measurements $A$,  as
$$OC(P_{Y|A,H_i}, P_{Y|A,H_j})=Q_{\lambda, i\rightarrow j}=Q_{\lambda, j\rightarrow i},$$
where $\lambda$ is chosen such that $Q_{\lambda, i\rightarrow j}=Q_{\lambda, j\rightarrow i}$.

Then with $O(\frac{k \log(n)}{\min\limits_{1\leq i,j \leq L, i\neq j} OC(P_{Y|A,H_i}, P_{Y|A,H_j})})$ random time-varying measurements, with high probability, Algorithms \label{alg:likelihoodratiotest_timevarying} and  \label{alg:likelihoodratiotest_timevarying} correctly identify the $k$ anomalous random variables.
Here $L$ is the number of hypothesis, $P_{Y|A,H_i}$, $1\leq i \leq n$ is the output probability distribution for measurements $Y$ under hypothesis $H_i$ and measurements $A$, and $OC(P_{Y|A,H_i}, P_{Y|A,H_j})$ is the outer conditional Chernoff information.

Moreover, the outer conditional Chernoff information  is also equal to
\begin{eqnarray*}
&&OC(P_{Y|A,H_i}, P_{Y|A,H_j})\\
&&=-\min_{0\leq \lambda \leq 1} \int P(A)\log ( \\
&&~~~~~~~~~~~~~~~~~~~~\int{P_{Y|A,H_i}^{\lambda}(x)P_{Y|A,H_j}^{1-\lambda}(x)\,dx}) \,dA\\
&&=-\min_{0\leq \lambda \leq 1}  E_{A} \left(\log \left(\int{P_{Y|A,H_i}^{\lambda}(x)P_{Y|A,H_j}^{1-\lambda}(x)\,dx}\right)\right)
\end{eqnarray*}
\label{thm:numberofsamples_deterministic_timevarying}
\end{theorem}

\begin{proof}
In Algorithm \ref{alg:timevaryingalg2deterministic}, for two different hypothesis $H_i$ and $H_j$, we choose the probability likelihood ratio threshold of the Neyman-Pearson testing in a way, such that the hypothesis testing error probability decreases with the largest error exponent. Now we focus on deriving what this largest error exponent is, under deterministic time-varying measurements.

For simplicity of presentation, we first consider a special case: there are only two possible measurements $A_1$ and $A_2$; and one half of the measurements are $A_1$ while the other half are $A_2$. The conclusions can be extended to general distribution $P(A)$ on $A$, in a similar way of reasoning. Suppose we take $m$ measurements in total, our assumption translates to that $\frac{1}{2} m$ measurements are taken as $A_1$, and $\frac{1}{2} m$ measurements are taken as $A_2$. Without loss of generality, we consider two hypothesis denoted by $H_1$ and $H_2$. Under measurement $A_1$, we assume that $H_1$ generates distribution $P_1$ for observation data; $H_2$ generates distribution $P_2$ for observation data. Under $A_2$, we assume that $H_1$ generates distribution $P_3$ for observation data; $H_2$ generates distribution $P_4$ for observation data. In addition, we assume that the observation data is over a discrete space $\chi$, which can also be generalized to a continuous space without affecting the conclusion in this theorem.

Suppose that $P$ is the empirical distribution of measurement data under measurements $A_1$, and that $P'$ is the empirical distribution of measurement data under measurements $A_2$. Then the Neyman-Pearson testing decides that hypothesis $H_1$ is true if, for a certain constant $T$, $$\frac{1}{2}[D(P||P_2)-D(P||P_1)]+\frac{1}{2}[D(P'||P_4)-D(P'||P_3)]\geq \frac{1}{n} \log(T). $$
By the Sanov's theorem \cite{coverbook}, the error exponent of the second kind, namely wrongly deciding ``hypothesis $H_1$ is true'' when hypothesis $H_2$ is actually true, is given by the following optimization problem.
$$ \min_{P,P'} \frac{1}{2} D(P||P_2)+\frac{1}{2} D(P'||P_4)$$
 subject to
 $$\frac{1}{2}[D(P||P_2)-D(P||P_1)]+\frac{1}{2}[D(P'||P_4)-D(P'||P_3)]\geq \frac{1}{n} \log(T) $$
 $$\sum_{x} P(x)=1$$
 $$\sum_{x} P'(x)=1$$

Using the Lagrange multiplier method, we try to minimize
\begin{eqnarray}
&&\frac{D(P||P_2)}{2} +\frac{D(P'||P_4)}{2} +\lambda \left( \frac{D(P||P_2)-D(P||P_1)}{2}\right.\nonumber\\
&&\left.+\frac{D(P'||P_4)-D(P'||P_3)}{2} \right )\nonumber\\
&&+ v_1 \sum_{x} P(x)+v_2 \sum_{x} P'(x) \nonumber
\end{eqnarray}

Differentiating with respect to $P(x)$ and $P'(x)$, we get
$$\frac{1}{2}[\log(\frac{P(x)}{P_2(x)})+1+\lambda \log (\frac{P_1(x)}{P_2(x)})]+v_1=0$$
$$\frac{1}{2}[\log(\frac{P'(x)}{P_4(x)})+1+\lambda \log (\frac{P_3(x)}{P_4(x)})]+v_2=0$$

From these equations, we can obtain the minimizing $P$,
$$P=P_{\lambda}(Y|A_1)=\frac{P_1^\lambda(x)P_2^{1-\lambda}(x)}{\sum_{x \in \chi}{P_1^\lambda(x)P_2^{1-\lambda}(x)} }$$
$$P'=P_{\lambda}(Y|A_2)=\frac{P_3^\lambda(x)P_4^{1-\lambda}(x)}{\sum_{x \in \chi}{P_3^\lambda(x)P_4^{1-\lambda}(x)} },$$
where $\lambda$ is chosen such that $\frac{1}{2}[D(P||P_2)-D(P||P_1)]+\frac{1}{2}[D(P'||P_4)-D(P'||P_3)]= \frac{1}{n} \log(T)$.

By symmetry, the error exponent of the second kind, and the error exponent of the first kind, are respectively
$$\frac{1}{2} D(P_{\lambda}(Y|A_1)||P_2)+\frac{1}{2} D(P_{\lambda}(Y|A_2)||P_4)$$
$$\frac{1}{2} D(P_{\lambda}(Y|A_1)||P_1)+\frac{1}{2} D(P_{\lambda}(Y|A_2)||P_3)$$

The first exponent is an increasing function in $\lambda$, and the second exponent is a decreasing function in $\lambda$. In fact, the optimal error exponent, which is the minimum of these two exponents, is achieved when they are equal:
\begin{eqnarray*}
&~&\frac{1}{2} D(P_{\lambda}(Y|A_1)||P_2)+\frac{1}{2} D(P_{\lambda}(Y|A_2)||P_4)\\
&=&\frac{1}{2} D(P_{\lambda}(Y|A_1)||P_1)+\frac{1}{2} D(P_{\lambda}(Y|A_2)||P_3)
\end{eqnarray*}
This finishes the characterization of the optimal error exponent in pairwise hypothesis testing, under deterministic time-varying measurements.

Based on this result, we further show that the derived optimal exponent is equivalent to
\begin{eqnarray}
&&OC(P_{Y|A,H_i}, P_{Y|A,H_j})\nonumber\\
&&=-\min_{0\leq \lambda \leq 1} \int P(A) \log\left(\int{P_{Y|A,H_i}^{\lambda}(x)P_{Y|A,H_j}^{1-\lambda}(x)\,dx} \right)\,dA \nonumber\\
&&=-\min_{0\leq \lambda \leq 1}  E_{A} \left( \log \left(\int{P_{Y|A,H_i}^{\lambda}(x)P_{Y|A,H_j}^{1-\lambda}(x)\,dx}\right) \right)\nonumber\\
&&\leq E_{A} \{C(P_{Y|A,H_i}, P_{Y|A, H_j})\}.
\label{eq:conditionalchernoff inequality}
\end{eqnarray}
In this proof, we restrict our attention to $H_i=H_1$ and $H_j=H_2$.

We first show that the $0 \leq \lambda \leq 1$ that minimizes $E_{A} \left( \log \left(\int{P_{Y|A,H_i}^{\lambda}(x)P_{Y|A,H_j}^{1-\lambda}(x)\,dx}\right) \right)$ is exactly $\lambda$ which leads to the equality:
\begin{eqnarray}
&~&\frac{1}{2} D(P_{\lambda}(Y|A_1)||P_2)+\frac{1}{2} D(P_{\lambda}(Y|A_2)||P_4) \nonumber\\
&=&\frac{1}{2} D(P_{\lambda}(Y|A_1)||P_1)+\frac{1}{2} D(P_{\lambda}(Y|A_2)||P_3)
\label{eq:optimallambda}
\end{eqnarray}

On the one hand, this equality means that
\begin{eqnarray}
0&=&\frac{1}{2} [D(P_{\lambda}(Y|A_1)||P_2)-D(P_{\lambda}(Y|A_1)||P_1)] \nonumber \\
&~& +\frac{1}{2} [D(P_{\lambda}(Y|A_2)||P_4)-D(P_{\lambda}(Y|A_2)||P_3)]\nonumber\\
&=&\frac{1}{2} \frac{\sum_{x}{P_1^{\lambda}(x)P_2^{1-\lambda}(x) \log(\frac{P_2(x)}{P_1(x)})} }{ \sum_{x} P_1^\lambda(x)P_2^{1-\lambda}(x)}\nonumber\\
&~& +\frac{1}{2} \frac{\sum_{x}{P_3^{\lambda}(x)P_4^{1-\lambda}(x) \log(\frac{P_4(x)}{P_3(x)})} }{ \sum_{x} P_3^\lambda(x)P_4^{1-\lambda}(x)}\nonumber
\end{eqnarray}

Moreover, under this $\lambda$, the exponent is equal to,
\begin{eqnarray*}
&&\frac{1}{2} D(P_{\lambda}(Y|A_1)||P_1)+\frac{1}{2} D(P_{\lambda}(Y|A_2)||P_3) \\
&=& \frac{1}{2} \frac{1}{\sum_{x} P_1^\lambda(x)P_2^{1-\lambda}(x)}
\sum_{x}P_1^{\lambda}(x)P_2^{1-\lambda}(x) \\
&~&[(1-\lambda)\log(\frac{P_2(x)}{P_1(x)})-\log(\sum_{x} P_1^{\lambda}(x)P_2^{1-\lambda}(x))]\\
&~& +\frac{1}{2} \frac{1}{\sum_{x} P_3^\lambda(x)P_4^{1-\lambda}(x)}\sum_{x}P_3^{\lambda}(x)P_4^{1-\lambda}(x) \\
&~&[(1-\lambda)\log(\frac{P_4(x)}{P_3(x)})-\log(\sum_{x} P_3^{\lambda}(x)P_4^{1-\lambda}(x))]
\end{eqnarray*}

Recognizing the first parts of both summations over $x$ is equal to $0$, we have
\begin{eqnarray*}
&&\frac{1}{2} D(P_{\lambda}(Y|A_1)||P_1)+\frac{1}{2} D(P_{\lambda}(Y|A_2)||P_3) \\
&=&-\frac{1}{2} \log(\sum_{x} P_1^{\lambda}(x)P_2^{1-\lambda}(x))\\
&~&-\frac{1}{2} \log(\sum_{x} P_3^{\lambda}(x)P_4^{1-\lambda}(x)),
\end{eqnarray*}
which is just the the expression,  under $\lambda^*$ achieving \ref{eq:optimallambda}, of the second definition of the outer conditional Chernoff information.

On the other hand, to minimize
$$E_{A} \left( \log \left(\int{P_{Y|A,H_i}^{\lambda}(x)P_{Y|A,H_j}^{1-\lambda}(x)\,dx}\right) \right),$$
the derivative with respect to $\lambda$ is equal to $0$:
\begin{eqnarray*}
0&=&\frac{1}{2} \frac{\sum_{x}{P_1^{\lambda}(x)P_2^{1-\lambda}(x) \log(\frac{P_2(x)}{P_1(x)})} }{ \sum_{x} P_1^\lambda(x)P_2^{1-\lambda}(x)}\\
&~& +\frac{1}{2} \frac{\sum_{x}{P_3^{\lambda}(x)P_4^{1-\lambda}(x) \log(\frac{P_4(x)}{P_3(x)})} }{ \sum_{x} P_3^\lambda(x)P_4^{1-\lambda}(x)}.
\end{eqnarray*}
This means that the optimizing $\lambda_{min}=\lambda^*$. So $E_{A} \left( \log \left(\int{P_{Y|A,H_i}^{\lambda}(x)P_{Y|A,H_i}^{1-\lambda}(x)\,dx}\right) \right)=\frac{1}{2} D(P_{\lambda}(Y|A_1)||P_1)+\frac{1}{2} D(P_{\lambda}(Y|A_2)||P_3)$, where $\lambda^*$ satisfies
\begin{eqnarray}
&~&\frac{1}{2} D(P_{\lambda}(Y|A_1)||P_2)+\frac{1}{2} D(P_{\lambda}(Y|A_2)||P_4) \nonumber\\
&=&\frac{1}{2} D(P_{\lambda}(Y|A_1)||P_1)+\frac{1}{2} D(P_{\lambda}(Y|A_2)||P_3)
\label{eq:optimallambda}
\end{eqnarray}

%
%
%
%
Overall, the smallest possible error exponent between any pair of hypothesis is
\begin{equation*}
E=\min_{1\leq i,j \leq L, i\neq j} OC(P_{Y|A,H_i}, P_{Y|A,H_j}).
\end{equation*}

Without loss of generality, we assume $H_1$ is the true hypothesis. Since the error probability $P_{e}$ in the Neyman-Pearson testing is
$$P_{e}\doteq2^{-m OC(P_{Y|A,H_i}, P_{Y|A,H_j})} \leq 2^{-mE}.$$ By a union bound over the $L-1$ possible pairs $(H_1, H_j)$, the probability that $H_1$ is not correctly identified as the true hypothesis is upper bounded by $L2^{-mE}$ in terms of scaling. So $m=\Theta(k \log(n)E^{-1})$ samplings are enough for identifying the $k$ anomalous samples with high probability.  When $E$ grows polynomially with $n$, this implies a significant reduction in the number of samples needed.

\end{proof}
\section{Examples of Mixed Observations Strictly Reducing Hypothesis Testing Error Probability}
\label{sec:example}

In this section, we give examples in which smaller error probability can be achieved in hypothesis testing through mixed observations. In particular, we consider Gaussian distributions in our examples.

\subsection{Example 1: two Gaussian random variables}
In this example, we consider $n=2$, and $k=1$. We group the two random variables $X_1$ and $X_2$ in a random vector $(X_1, X_2)$. Suppose that there are two hypothesis for a $2$-dimensional random vector $(X_1, X_2)$:
\begin{itemize}
 \item $H_1$: $X_1$ and $X_2$ are independent random variables; $X_1$ and $X_2$ follow Gaussian distributions $\mathcal{N}(A, \sigma^2)$ and $\mathcal{N}(B, \sigma^2)$ respectively.
 \item $H_2$: $X_1$ and $X_2$ are independent random variables; $X_1$ and $X_2$ follow Gaussian distributions $\mathcal{N}(B, \sigma^2)$ and $\mathcal{N}(A, \sigma^2)$ respectively.
\end{itemize}
Here $A$ and $B$ are two distinct constants, and $\sigma^2$ is the variance of the two Gaussian random variables. At each time instant, only one observation is allowed, and the observation is restricted to a linear mixing of $X_1$ and $X_2$. Namely
$$Y_j=\alpha_1 X_1+\alpha_2 X_2.$$
We assume that the linear mixing does not change over time.

Clearly, when $\alpha_1 \neq 0$ and $\alpha_2=0$, this linear mixing reduces to a separate observation of $X_1$; and when $\alpha_1 = 0$ and $\alpha_2 \neq 0$, it reduces to a separate observation of $X_2$. In both these two cases, the observation follows distribution $\mathcal{N}(A,\sigma^2)$ for one hypothesis, and follows distribution $\mathcal{N}(B, \sigma^2)$ for the other hypothesis. The Chernoff information between these two distributions are
$$C(\mathcal{N}(A, \sigma^2),\mathcal{N}(B, \sigma^2))=\frac{(A-B)^2}{8\sigma^2}$$

When the hypothesis $H_1$ holds, the observation $Y_j$ follows the distribution $\mathcal{N}(\alpha_1 A+\alpha_2 B, (\alpha_1^2+\alpha_2^2)\sigma^2 )$. Similarly, when the hypothesis $H_2$ holds, the observation $Y_j$ follows the distribution $\mathcal{N}(\alpha_1 B+\alpha_2 A, (\alpha_1^2+\alpha_2^2)\sigma^2 )$.  The Chernoff information between these two Gaussian distributions $\mathcal{N}(\alpha_1 A+\alpha_2 B, (\alpha_1^2+\alpha_2^2)\sigma^2 )$, and $\mathcal{N}(\alpha_1 B+\alpha_2 A, (\alpha_1^2+\alpha_2^2)\sigma^2 )$, is given by
\begin{eqnarray*}
&&\frac{[(\alpha_1 A+\alpha_2 B)-(\alpha_1 B+\alpha_2 A)]^2}{8(\alpha_1^2+\alpha_2^2)\sigma^2}\\
&&=\frac{[(\alpha_1-\alpha_2)^2(A-B)^2]}{8(\alpha_1^2+\alpha_2^2)\sigma^2}\\
&&\leq \frac{2(A-B)^2}{8\sigma^2},
\end{eqnarray*}
where the last inequality follows from Cauchy-Schwarz inequality, and takes equality when $\alpha_1=-\alpha_2$.

Compared with the Chernoff information for separate observations of $X_1$ or $X_2$, the linear mixing of $X_1$ and $X_2$ doubles the Chernoff information. This shows that linear mixed observations can offer strict improvement in terms of reducing the error probability in hypothesis testing, by increasing the error exponent, compared with separate observations of random variables.

\subsection{Example 2: Gaussian random variables with different means}
In this example, we consider the mixed observations for two hypothesis of Gaussian random vectors. In general, suppose that there are two hypothesis for an $N$-dimensional random vector $(X_1, X_2,...,X_{N})$,
\begin{itemize}
 \item $H_1$: $(X_1, X_2,...,X_{N})$ follow jointly Gaussian distributions $\mathcal{N}(\mu_1, \Sigma_1)$.
 \item $H_2$: $(X_1, X_2,...,X_{N})$ follow jointly Gaussian distributions $\mathcal{N}(\mu_2, \Sigma_2)$.
\end{itemize}
Here $\Sigma_1$ and $\Sigma_2$ are both $N \times N$ covariance matrices.

At each time instant, only one observation is allowed, and the observation is restricted to a time-invariant linear mixing of $X=(X_1, X_2,...,X_{N})$. Namely
$$Y_j=<A, X>.$$

Under these conditions, the observation follows distribution $\mathcal{N}(A^T \mu_1, A^T \Sigma_1 A)$ for the hypothesis $H_1$, and follows distribution $\mathcal{N}(A^T \mu_2, A^T \Sigma_2 A)$ for the other hypothesis $H_2$. We would like to choose a linear mixing $A$ which maximizes the Chernoff information between the two possible univariate Gaussian distributions, namely
$$\max_{A} C(\mathcal{N}(A^T \mu_1, A^T \Sigma_1 A),\mathcal{N}(A^T \mu_2, A^T \Sigma_2 A) )$$

In fact, the Chernoff information between these two distributions are \cite{ChernoffInformation}
\begin{eqnarray*}
&~&C(\mathcal{N}(A^T \mu_1, A^T \Sigma_1 A),\mathcal{N}(A^T \mu_2, A^T \Sigma_2 A) )\\
&=&\max_{0\leq \alpha \leq 1} \left[ \frac{1}{2} \log{\left( \frac{A^T (\alpha \Sigma_1+(1-\alpha) \Sigma_2) A}{(A^T \Sigma_1 A)^{\alpha} (A^T \Sigma_2 A)^{1-\alpha}} \right)} \right.\\
&~&~~~~~~\left.+\frac{\alpha (1-\alpha) (A^T (\mu_1-\mu_2))^2}{2 A^T (\alpha \Sigma_1+(1-\alpha) \Sigma_2) A} \right]
\end{eqnarray*}

We first look at the special case when $\Sigma=\Sigma_1=\Sigma_2$. Under this condition, the maximum Chernoff information is given by
$$\max_{A}\max_{0\leq \alpha \leq 1} \frac{\alpha(1-\alpha) [A^T (\mu_1-\mu_2)]^2}{2A^T \Sigma A}.$$

Taking $A'=\Sigma^{\frac{1}{2}} A$, then this reduces to
$$\max_{A}\max_{0\leq \alpha \leq 1} \frac{\alpha(1-\alpha) [(A')^T \Sigma^{-\frac{1}{2}}(\mu_1-\mu_2)]^2}{2(A')^T A'}.$$

From Cauchy-Schwarz inequality, it is easy to see that the optimal $\alpha=\frac{1}{2}$, $A'=\Sigma^{-\frac{1}{2}}(\mu_1-\mu_2)$, and $A=\Sigma^{-1}(\mu_1-\mu_2)$. Under these conditions, the maximum Chernoff information is given by
$$\frac{1}{8} (\mu_1-\mu_2)^T \Sigma^{-1} (\mu_1-\mu_2).$$
Note that in general, $A'=\Sigma^{-\frac{1}{2}}(\mu_1-\mu_2)$ is not a separate observation of a certain individual random variable, but rather a linear mixing of the $N$ random variables.

\subsection{Example 3: Gaussian random variables with different variances}
In this example, we look at the mixed observations for Gaussian random variables with different variances.  Consider the same setting in Example 2, except that we now look at the special case when $\mu=\mu_1=\mu_2$. We will study the optimal linear mixing under this scenario. Then the Chernoff information becomes
\begin{eqnarray*}
&~&C(\mathcal{N}(A^T \mu, A^T \Sigma_1 A),\mathcal{N}(A^T \mu, A^T \Sigma_2 A) )\\
&=&\max_{0\leq \alpha \leq 1} \frac{1}{2} \log{\left( \frac{A^T (\alpha \Sigma_1+(1-\alpha) \Sigma_2) A}{(A^T \Sigma_1 A)^{\alpha} (A^T \Sigma_2 A)^{1-\alpha}} \right)}
\end{eqnarray*}

To find the optimal projection $A$, we are solving this optimization problem
$$\max_{A}\max_{0\leq \alpha \leq 1} \frac{1}{2} \log{\left( \frac{A^T (\alpha \Sigma_1+(1-\alpha) \Sigma_2) A}{(A^T \Sigma_1 A)^{\alpha} (A^T \Sigma_2 A)^{1-\alpha}} \right)}.$$

For a certain $A$, we define
$$B=\frac{\max{(A^T \Sigma_1 A, A^T \Sigma_2 A) }}{\min{(A^T \Sigma_1 A, A^T \Sigma_2 A)  }}.$$

By symmetry over $\alpha$ and $1-\alpha$, maximizing the Chernoff information can always be reduced to
$$\max_{A} \max_{0\leq \alpha \leq 1} \frac{1}{2} \log{\left( \frac{\alpha + (1-\alpha) B}{B^{1-\alpha}} \right)}.$$

The optimal value is given by
$$\max_{A}\left \{-\frac{1}{B} e^{[-1+\frac{B}{B-1}\log(B)]  }   \left( -B+ \frac{B-1+B\log(B)}{\log(B)} \right)    \right\}  $$
under the optimizing $\alpha$ given by
$$\alpha=\frac{-(B-1)+B\log(B)}{(B-1) \log(B)}.$$

We note that the optimal value is an increasing function of $B$.
\begin{lemma}
The optimal objective value of the following optimization problem
$$\max_{A} \max_{0\leq \alpha \leq 1} \frac{1}{2} \log{\left( \frac{\alpha + (1-\alpha) B}{B^{1-\alpha}} \right)},$$
is an increasing function in $B$.
\end{lemma}

\begin{proof}
 We only need to show that for any $\alpha \in [0,1]$, $\left( \frac{\alpha + (1-\alpha) B}{B^{1-\alpha}} \right)$ is an increasing function in $B\geq 1$.
In fact, the derivative of it with respect to $\alpha$ is
$$\alpha(1-\alpha) (B^{\alpha-1}-B^{\alpha-2})\geq 0.$$
Then the conclusion of this lemma immediately follows.
\end{proof}

This means we need to maximize $B$, in order to maximize the Chernoff information. So to find the optimal $A$ to maximize the Chernoff information, we solve the following two optimization problems.
\begin{equation}
       \max_{A}{A^T \Sigma_1 A } ~~~~~\text{subject~to}~~A^T \Sigma_2 A \leq 1;
\label{eq:twoquadraticprogram1}
\end{equation}
and
\begin{equation}
       \max_{A}{A^T \Sigma_2 A } ~~~~~\text{subject~to}~~A^T \Sigma_1 A \leq 1;
\label{eq:twoquadraticprogram2}
\end{equation}
Then the maximum of the two optimal objective values is equal to the optimizing $B$, and the corresponding optimizing $A$ is the optimal linear projection we are looking for. These two optimization problems are not convex optimization programs, however, they still admit zero duality gap from from the S-procedure, and can be efficiently solved \cite{boyd}. In fact, they are respectively equivalent to the following two semidefinite programming optimization problems, and thus can be efficiently solved.
\begin{equation}
\begin{aligned}
& \underset{\gamma, \lambda}{\text{minimize}}
& & -\gamma \\
& \text{subject to}
& & \lambda \geq 0 \\
&&& {\left( \begin{array}{cc}
-\Sigma_1+\lambda\Sigma_2 & 0\\
0 & -\lambda-\gamma  \end{array} \right) \succeq 0}.
\end{aligned}
\label{eq:sdpfor2quadratic1}
\end{equation}

\begin{equation}
\begin{aligned}
& \underset{\gamma, \lambda}{\text{minimize}}
& & -\gamma \\
& \text{subject to}
& & \lambda \geq 0 \\
&&& {\left( \begin{array}{cc}
-\Sigma_2+\lambda\Sigma_1 & 0\\
0 & -\lambda-\gamma  \end{array} \right) \succeq 0}.
\end{aligned}
\label{eq:sdpfor2quadratic2}
\end{equation}

\subsection{Example 4: $k=1$ anomalous random variable among $n=7$ random variables }

Consider another example, where $k=1$ and $n=7$. $n-k=6$ random variables follow the distribution $\mathcal{N} (0, 1)$; and the other random variable follows distribution $\mathcal{N} (0, \sigma^2)$, where $\sigma^2\neq 1$.
So overall, there are $7$ hypothesis:
\begin{itemize}
 \item $H_1$: $X_1$, $X_2$,..., and $X_7$ are independent random variables; $X_1$, $X_2$,..., and $X_7$ follow Gaussian distributions $\mathcal{N}(0, \sigma^2)$, $\mathcal{N}(0, 1)$,..., and $\mathcal{N}(0, 1)$ respectively.
 \item $H_2$: $X_1$, $X_2$,..., and $X_7$ are independent random variables; $X_1$, $X_2$,..., and $X_7$ follow Gaussian distributions $\mathcal{N}(0, 1)$, $\mathcal{N}(0, \sigma^2)$,..., and $\mathcal{N}(0, 1)$ respectively.
 \item ...
 \item $H_7$: $X_1$, $X_2$,..., and $X_7$ are independent random variables; $X_1$, $X_2$,..., and $X_7$ follow Gaussian distributions $\mathcal{N}(0, 1)$, $\mathcal{N}(0, 1)$,... and $\mathcal{N}(0, \sigma^2)$ respectively.
\end{itemize}

We first assume that separation observations of these $7$ random variables are made. For any pair of hypothesis $H_i$ and $H_j$, the probability distributions for the output are respectively $\mathcal{N}(0, \sigma^2)$ and $\mathcal{N}(0, 1)$, when $X_i$ is observed; the probability distributions for the output are respectively $\mathcal{N}(0, 1)$ and $\mathcal{N}(0, \sigma^2)$, when $X_j$ is observed. From \ref{eq:conditionalchernoff inequality}, the Chernoff information between $H_i$ and $H_j$ is given by
\begin{eqnarray*}
&&OC(P_{Y|A,H_i}, P_{Y|A, H_j})\\
&&=-\min_{0\leq \lambda \leq 1}  E_{A} \left( \log \left(\int{P_{Y|A,H_i}^{\lambda}(x)P_{Y|A,H_i}^{1-\lambda}(x)\,dx}\right) \right)\\
&&=-\min_{0\leq \lambda \leq 1}  \left [\frac{1}{7} \left( \log \left(\int{P_{\mathcal{N}(0,1)}^{\lambda}(x)P_{\mathcal{N}(0,\sigma^2)}^{1-\lambda}(x)\,dx}\right) \right)\right. \\
&&~~~~\left. +\frac{1}{7} \left( \log \left(\int{P_{\mathcal{N}(0,\sigma^2)}^{\lambda}(x)P_{\mathcal{N}(0,1)}^{1-\lambda}(x)\,dx}\right) \right)\right]
\label{eq:differentvariancesSeparateObservations}
\end{eqnarray*}

Optimizing over $\lambda$, we obtain the optimal $\lambda=\frac{1}{2}$, and that
\begin{eqnarray}
OC(P_{Y|A,H_i}, P_{Y|A, H_j})=\frac{1}{7} \log(\frac{B+1}{2B^{\frac{1}{2}}}),
\label{eq:differentvariancesSeparateObservations_finalresults}
\end{eqnarray}
where $B=\frac{\max{(\sigma^2, 1)}}{\min{(\sigma^2, 1)}}$.

Now we consider using the parity check matrix of $(7,4)$ Hamming codes to do the measurements. For any pair of hypothesis $H_i$ and $H_j$, there is always a measurement row which measure one and only one of $X_i$ and $X_j$. Without loss of generality, we assume that that measurement measures $X_i$ but not $X_j$.
Suppose $H_i$ is true, then the output probability distribution for that measurement is $\mathcal{N}(0, \sigma^2+A)$; otherwise when $H_j$ is true,  the output probability distribution is given by $\mathcal{N}(0, 1+A)$, where $1+A$ is the number of ones in that measurement row. From \ref{eq:conditionalchernoff inequality},  the Chernoff information between $H_i$ and $H_j$ is bounded by
\begin{eqnarray*}
&&OC(P_{Y|A,H_i}, P_{Y|A, H_j})\\
&&=-\min_{0\leq \lambda \leq 1}  E_{A} \left( \log \left(\int{P_{Y|A,H_i}^{\lambda}(x)P_{Y|A,H_i}^{1-\lambda}(x)\,dx}\right) \right)\\
&&\geq -\min_{0\leq \lambda \leq 1}  \frac{1}{3}  \log \left(\int{P_{\mathcal{N}(0,\sigma^2+A)}^{\lambda}(x)P_{\mathcal{N}(0,1+A)}^{1-\lambda}(x)\,dx}\right)
\label{eq:differentvariancesHamming}
\end{eqnarray*}
Again this lower bound is given by
$$\frac{1}{3}\left \{-\frac{1}{B} e^{[-1+\frac{B}{B-1}\log(B)]  }   \left( -B+ \frac{B-1+B\log(B)}{\log(B)} \right)    \right\},$$
where $$B=\frac{\max{(\sigma^2+A, 1+A)}}{\min{(\sigma^2+A, 1+A)}}.$$

Simply taking $\lambda=\frac{1}{2}$, we get another lower bound of $C(P_{Y|A,H_i}, P_{Y|A, H_j})$:
$$\frac{1}{3}\times \frac{1}{2} \log\left(\sqrt{\frac{\sigma^2+A}{4(1+A)}}+\sqrt{\frac{1+A}{4(\sigma^2+A)}}\right).$$

When $\sigma^2 \gg 1$, for separate observations,
\begin{eqnarray}
OC(P_{Y|A,H_i}, P_{Y|A, H_j})\sim \frac{1}{14} \log(\sigma^2);
\label{eq:differentvariancesSeparateObservations_scaling}
\end{eqnarray}
while for measurements through the parity-check matrix of Hamming codes, a lower bound $\underline{O(P_{Y|A,H_i}, P_{Y|A, H_j})}$ of $OC(P_{Y|A,H_i}, P_{Y|A, H_j})$ asymptotically satisfies
\begin{eqnarray}
\underline{OC(P_{Y|A,H_i}, P_{Y|A, H_j})}\sim \frac{1}{12} \log{(\sigma^2)}.
\label{eq:differentvariancesHamming_scaling}
\end{eqnarray}
So in the end, the minimum Chernoff information between any pair of hypothesis, under mixed measurements using Hamming codes, is bigger than the Chernoff information obtained using separate observations. This means that mixed observations can offer strict improvement in the error exponent of hypothesis testing problems.

\section{Characterization of Optimal Measurements Maximizing the Error Exponent}
\label{sec:optimalmixing}

In this section, we derive a characterization of the optimal deterministic time-varying measurements which maximize the error exponent of hypothesis testing. We further explicitly design the optimal measurements for some simple examples. We begin with the following lemma about the error exponent of hypothesis testing.
\begin{lemma}
Suppose that there are overall $L=\binom{n}{k}$ hypothesis assumptions. For any fixed $k$ and $n$, the error exponent of the error probability of hypothesis testing is given by
\begin{equation*}
E=\min_{1\leq i,j \leq L, i\neq j} OC(P_{Y|A,H_i}, P_{Y|A,H_j}).
\end{equation*}
\label{lem:exactexponent}
\end{lemma}

\begin{proof}
We first give an upper bound on the error probability of hypothesis testing. Without loss of generality, we assume $H_1$ is the true hypothesis. Since the error probability $P_{e}$ in the Neyman-Pearson testing is
$$P_{e}\doteq2^{-m OC(P_{Y|A,H_i}, P_{Y|A, H_j})} \leq 2^{-mE}.$$
By a union bound over the $L-1$ possible pairs $(H_1, H_j)$, the probability that $H_1$ is not correctly identified as the true hypothesis is upper bounded by $L2^{-mE}$ in terms of scaling.

Now we give a lower bound on the error probability of hypothesis testing. Without loss of generality, we assume that $E$ is achieved between the hypothesis $H_1$ and the hypothesis $H_2$, namely,
\begin{equation*}
E=OC(P_{Y|A,H_1}, P_{Y|A,H_2}).
\end{equation*}

Suppose that we are given the prior information that either hypothesis $H_1$ or $H_2$ is true. Knowing this prior information will not increase the error probability. Under this prior information, the error probability behaves asymptotically as $2^{-m E}$ as $m\rightarrow \infty$.  This shows that the error exponent of hypothesis testing is no bigger than $E$.
\end{proof}

The following theorem gives a simple characterization of the optimal $p(A)$. This enables us to explicitly find the optimal mixed measurements, under certain special cases of Gaussian random variables.
\begin{theorem}
In order to maximize the error exponent in hypothesis testing, the optimal mixed measurements have a distribution $p^*(A)$ which maximizes the minimum of the pairwise outer Chernoff information between different hypothesis:
\begin{equation*}
P^*(A)=\mbox{arg}\max_{P(A)} \min_{1\leq i,j \leq L, i\neq j} OC(P_{Y|A,H_i}, P_{Y|A,H_j}).
\end{equation*}

When $k=1$ and the $n$ random variables of interest are independent Gaussian random variables with the same variances, the optimal $P^*(A)$ admit a discrete probability distribution:
\begin{equation*}
P(A)= \sum_{\sigma \mbox{as a permutation} } \frac{1}{n!}\delta(A-\sigma(A^*)),
\end{equation*}
where $A^*$ is a constant $n$-dimensional vector such that
$$A^*=\mbox{arg}\max_{A} \sum\limits_{1\leq i,j \leq L, i\neq j} C(P_{Y|A,H_i}, P_{Y|A,H_j}). $$
\end{theorem}

\begin{proof}
The first statement follows from Lemma \ref{lem:exactexponent}.  So we only need to prove the optimal mixed measurements for Gaussian random variables with the same variance under $k=1$.
For any $p(A)$, we have
\begin{eqnarray*}
&~& \min_{1\leq i,j \leq L, i\neq j} OC(P_{Y|A,H_i}, P_{Y|A,H_j}) \\
 &\leq& \frac{1}{\binom{n}{2}}\sum\limits_{1\leq i,j \leq L, i\neq j} OC(P_{Y|A,H_i}, P_{Y|A,H_j})\\
 &\leq&\frac{1}{\binom{n}{2}}\sum\limits_{1\leq i,j \leq L, i\neq j} \int_{A}{p(A) OC(P_{Y|A,H_i}, P_{Y|A,H_j})}\,d A \\
  &=&\frac{1}{\binom{n}{2}}\int_{A}{p(A) \sum\limits_{1\leq i,j \leq L, i\neq j}  OC(P_{Y|A,H_i}, P_{Y|A,H_j})}\,d A \\
 &\leq & \frac{1}{\binom{n}{2}}\sum\limits_{1\leq i,j \leq L, i\neq j} OC(P_{Y|A^*,H_i}, P_{Y|A^*,H_j}),
\end{eqnarray*}
where the second inequality follows from (\ref{eq:conditionalchernoff inequality}).

On the other hand, for two Gaussian distributions with the same variances, the optimizing $\lambda$ in (\ref{eq:conditionalchernoff inequality}) is always equal to $\frac{1}{2}$, no matter what
$P(A)$ is chosen.  So, by symmetry, when
\begin{equation*}
P(A)= \sum_{\sigma \mbox{as a permutation} } \frac{1}{n!}\delta(A-\sigma(A^*)),
\end{equation*}
for any two different hypothesis $H_i$ and $H_j$,
\begin{eqnarray*}
&~& OC(P_{Y|A,H_i}, P_{Y|A,H_j}) \\
 &=&\int_{A}{p(A) OC(P_{Y|A,H_i}, P_{Y|A,H_j})}\,d A \\
  &=&\frac{1}{\binom{n}{2}}\int_{A}{p(A) \sum\limits_{1\leq i,j \leq L, i\neq j}  OC(P_{Y|A,H_i}, P_{Y|A,H_j})}\,d A \\
 &= & \int_{A}{p(A) \frac{1}{\binom{n}{2}} \sum\limits_{1\leq i,j \leq L, i\neq j}  OC(P_{Y|A^*,H_i}, P_{Y|A^*,H_j})}\,d A \\
 &=& \frac{1}{\binom{n}{2}} \sum\limits_{1\leq i,j \leq L, i\neq j}  OC(P_{Y|A^*,H_i}, P_{Y|A^*,H_j}).
\end{eqnarray*}

This means that
\begin{eqnarray*}
&&\min_{1\leq i,j \leq L, i\neq j} OC(P_{Y|A,H_i}, P_{Y|A,H_j})\\
&=&\frac{1}{\binom{n}{2}} \sum\limits_{1\leq i,j \leq L, i\neq j}  OC(P_{Y|A^*,H_i}, P_{Y|A^*,H_j}).
\end{eqnarray*}

Since the upper bound on $\min_{1\leq i,j \leq L, i\neq j} OC(P_{Y|A,H_i}, P_{Y|A,H_j})$ is achieved, we conclude that $p(A^*)$ is the optimizing distribution.
\end{proof}

\section{Simulation Results}
\label{sec:simulation}
In this section, we numerically evaluate the performance of mixed observations in hypothesis testing. We first simulate the error probability in identifying anomalous random variables through linear mixed observations.

 The linear mixing used in the first simulation is based on sparse bipartite graphs. In sparse bipartite graphs \cite{ExpanderCode, XHExpander, IndykSparse}, $n$ variable nodes on the left are used to represent the $n$ random variables, and $m$ measurement nodes on the right are used to represent the $m$ measurements. If and only if the $i$-th random variable is nontrivially involved in the $j$-th measurement, there is an edge connecting the $i$-th variable node to the $j$-th measurement node. Unlike sparse bipartite graphs already used in low-density parity-check codes, and compressed sensing \cite{ExpanderCode, XHExpander, IndykSparse}, a novelty in this paper is that our sparse bipartite graphs are allowed to have more measurement nodes than variable nodes, namely $m \geq n$. In this simulation, there are $6$ edges emanating from each measurement node on the right, and  there are $\frac{6m}{n}$ edges emanating from each variable node on the left. In our simulation,  after a uniformly random permutation, the $6m$ edges emanating from the measurements nodes are plugged into the $6m$ edge ``sockets'' of the left variable nodes. If there is an edge connecting connecting the $i$-th variable node to the $j$-th measurement node, then the linear mixing coefficient before the $i$-th random variable in the $j$-th measurement is set to $1$; otherwise that linear mixing coefficient is set to $0$.

For this simulation, we take $n=100$, and let $m$ vary from $50$ to $300$. $k=1$ random variable follows the Gaussian distribution $\mathcal{N}(0, 100)$, and the other $(n-k)=99$ random variables follow another distribution $\mathcal{N}(0,1)$. The likelihood ratio test algorithm was used to find the anomalous random variables through the described linear mixed observations based on sparse bipartite graphs. For comparison, we also implement the likelihood ratio test algorithms for separate observations of random variables, where we first make $\lfloor\frac{m}{n} \rfloor$ separate observations of each random variables, and then made an additional separate observation for uniformly randomly selected $(m\mod n)$ random variables.  For each $m$, we perform $1000$ random trials, and record the number of trials failing to identify the anomalous random variables.  The error probability, as a function of $m$, is plotted in Figure \ref{fig:N100Variance}. We can see that linear mixed observation offers significant reduction in the error probability of hypothesis testing, under the same number of observations.

Under the same simulation setup as in Figure \ref{fig:N100Variance}, we test the error probability performance of mixed observations for two Gaussian distributions: the anomalous Gaussian distribution $\mathcal{N}(0, 1)$, and the common Gaussian distribution $\mathcal{N}(8,1)$. We also slightly adjust the number of total random variables as $n=102$, to make sure that each random variable participates in the same integer number $\frac{6m}{n}$ of measurements. Mixed observations visibly reduces the error probability under the same number of measurements, compared with separate observations. For example, even when $m=68<n=102$, in $999$ out of $1000$ cases, the likelihood ratio test correctly identifies the anomalous random variable.

\begin{figure}[t]
\centering
\includegraphics[width=3.75in, height=2.5in]{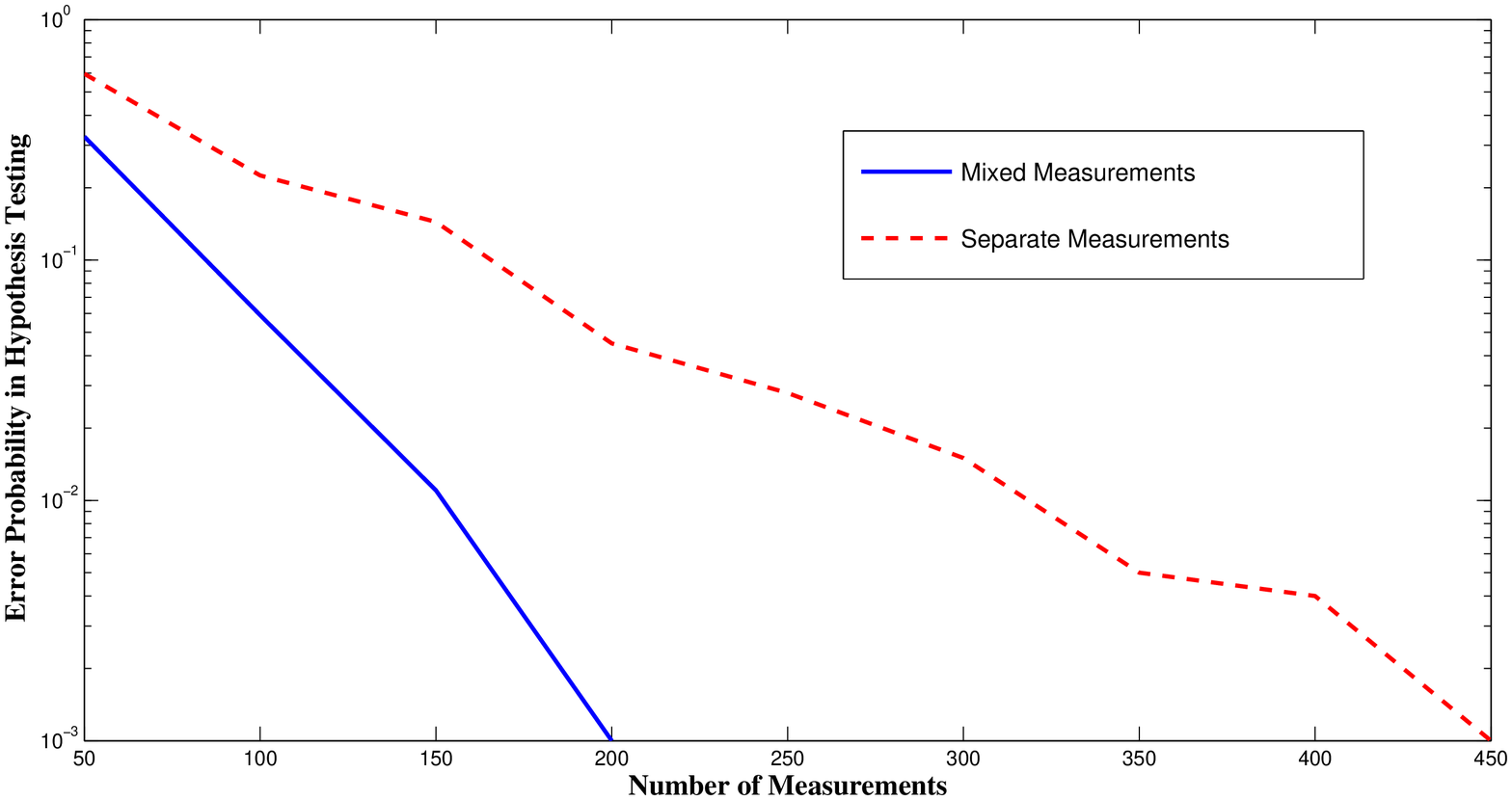}
\caption{Error probability versus $m$}\label{fig:N100Variance}
\end{figure}

\begin{figure}[t]
\centering
\includegraphics[width=3.75in, height=2.5in]{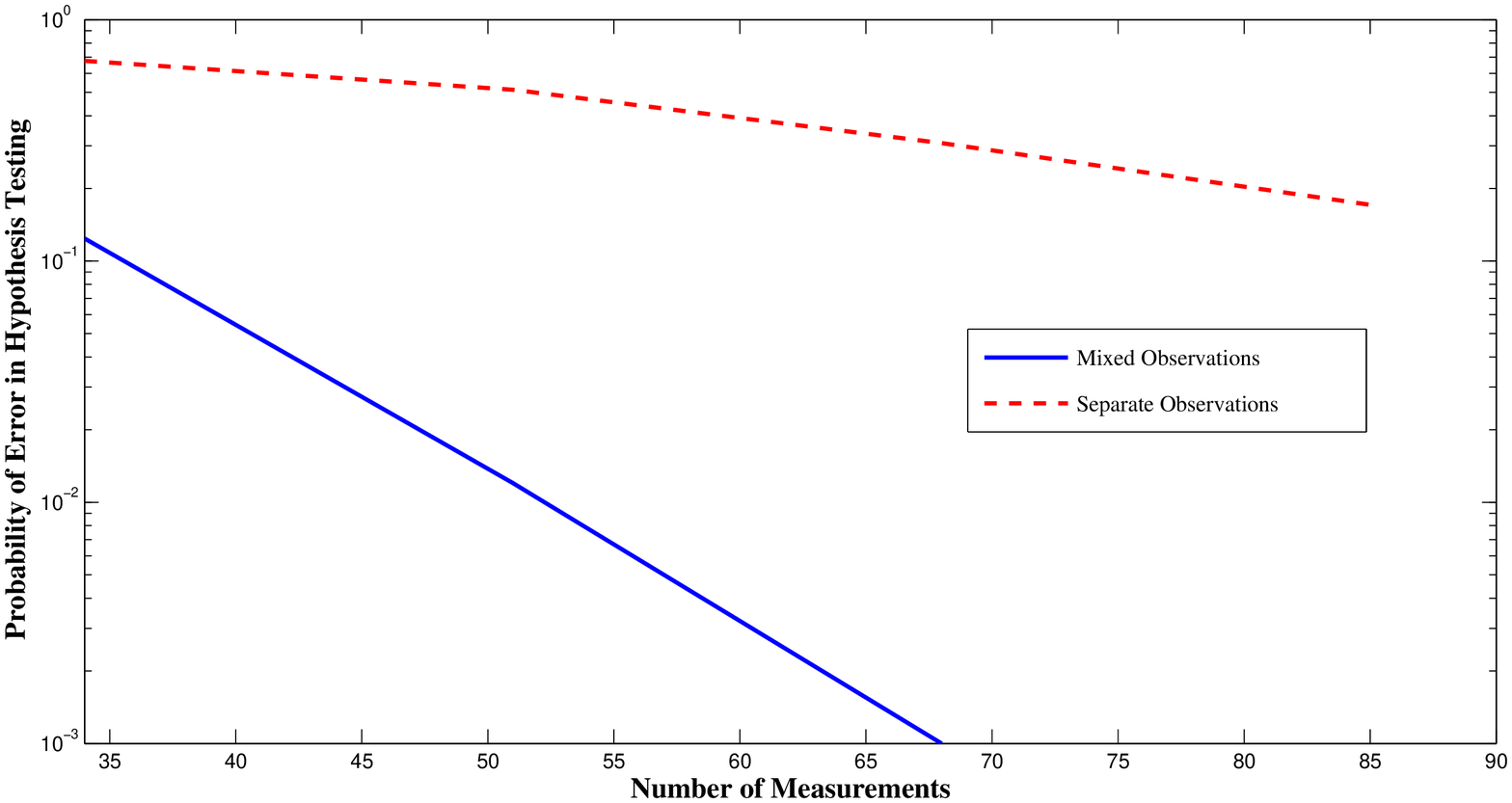}
\caption{Error probability versus $m$}\label{fig:N102Expectation}
\end{figure}

\section{Conclusion}\label{sec:conclusion}
In this paper, we have studied the problem of finding $k$ anomalous random variables following a different probability distribution among $n$ random variables, by using \emph{non-adaptive} mixed observations of these $n$ random variables. Our analysis has shown that mixed observations, compared with separate observations of individual random variables, can significantly reduce the number of samplings needed to identify the anomalous random variables. Compared with general compressed sensing problems, in our setting, each random variable may take dramatically different realizations in different observations.

There are some questions that remain open in performing hypothesis testing from mixed observations.
\begin{itemize}
\item What if the $n$ random variables are correlated random variables instead of independent random variables? Do mixed observations offer advantages?
\item What are the optimal constructions minimizing the number of mixed observations for general random variables?
\item What efficient algorithms are available to identify the $k$ anomalous random variables, using the probabilistic mixed observations, instead of the exhaustive search in this regime?
\item How are the hypothesis testing results from mixed observations affected by observation noises and errors?
\item What if the $k$ anomalous random variables are allowed to have different distributions? Can we further distinguish what distributions each anomalous  random variable follows?  We note that when we allow anomalous random variables to have different distributions, the traditional compressed sensing problem is just a special case of this hypothesis testing problem from mixed observations.
\end{itemize}
\bibliographystyle{IEEEbib}

\end{document}